\documentclass{article}

\usepackage[utf8]{inputenc}
\usepackage[T1]{fontenc}
\usepackage{amsmath,amssymb,amsfonts}
\usepackage{mathtools}
\usepackage{amsthm}
\usepackage{times}
\usepackage{enumitem}
\usepackage{xcolor}
\usepackage{hyperref}
\usepackage[margin=0.7in]{geometry}
\usepackage{tikz}
\usetikzlibrary{arrows.meta, positioning, calc,fit}

\newtheorem{theorem}{Theorem}
\newtheorem{lemma}{Lemma}

\newtheorem{corollary}{Corollary}
\newtheorem{assumption}{Assumption}
\newtheorem{definition}{Definition}
\newtheorem{remark}{Remark}

\newcommand{\R}[1][\empty]{\mathbb{R}^{#1}}
\newcommand{\Rp}[1][\empty]{\R[#1]_{\ge 0}}
\newcommand{\eps}{\varepsilon}

\newcommand{\xsas}{\overline x}
\newcommand{\ysas}{\overline y}

\allowdisplaybreaks

\newcommand{\xsol}{\mathbf{x}}

\title{On input-output persistency and the interconnection of positive nonlinear systems}

\author{Arthur Castello B. de Oliveira, 
Moh K. Wafi
, and Eduardo D. Sontag}

\date{}

\begin{document}

\maketitle

\begin{abstract}
We study feedback interconnections of positive nonlinear SISO systems and introduce two complementary properties: persistent-input/persistent-output (PIPO) for the plant and persistent-input/transient-output (PITO) for the controller. Assuming forward completeness of the closed loop and a one-sided affine growth bound on the controller output, we show that PIPO and PITO jointly imply boundedness of the control signal. The result is input-output in nature and does not require linearity or monotonicity of the interconnected subsystems, although it does not in general guarantee boundedness of the full controller state. We provide a structural PIPO condition for positive monotone plants with a class-\(\mathcal K_\infty\) steady-state characteristic, establish PITO for the antithetic integral controller with explicit gains, and illustrate the framework on a nonlinear integral-feedback motif with multiplicative controller growth.
\end{abstract}

\section{Introduction}
\label{sec:introduction}

Integral feedback is a canonical mechanism for robust regulation against
constant references and disturbances. When the closed loop is stable, the
integrator supplies the internal model required for asymptotic tracking and
disturbance rejection. This principle is fundamental in classical control
theory and has also become an important organizing concept in biological
regulation~\cite{rev_internal_model_2022,18cdc_tutorial_imp}. Integral action,
however, may interact poorly with hard constraints. In engineering systems,
actuator saturation may prevent the controller from applying the accumulated
control action while the integral state continues to grow, producing the
phenomenon known as windup. The analysis and compensation of windup have
therefore received considerable attention in the control
literature~\cite{galeani2009tutorial}. Biomolecular control systems face a
different but related constraint: concentrations and reaction rates must
remain nonnegative. Although positivity is not equivalent to actuator
saturation, it imposes a hard feasibility constraint and raises a similar
basic question: under what conditions do the signals generated by an integral
feedback interconnection remain bounded?

This question is particularly relevant for the antithetic integral controller
(AIC), introduced by Briat, Gupta, and Khammash in~\cite{AFC2016}. The AIC
implements integral action through an annihilation reaction between two
molecular species and has become a prominent feedback architecture in
synthetic biology. Under the assumptions considered in~\cite{Khammash2019},
the antithetic motif represents a fundamental biomolecular controller topology
capable of achieving robust perfect adaptation. The architecture has also been
realized experimentally. An in vivo implementation based on a
$\sigma^W$--RsiW regulatory pair was reported in~\cite{Khammash2019}, while
an in vitro controller based on $\sigma^{28}$ sequestration was used to
regulate gene expression in an \emph{E.\ coli} cell-free
transcription--translation system~\cite{agarwal2019naturecom,agarwal_cdc2019}.
Related quasi-integral, exponential, and logistic positive controllers have
also been proposed and analyzed~\cite{Huang2018,corentin_2020}.

The dynamics of antithetic feedback systems can be substantially more
complicated than their integral-control interpretation might suggest. For the
standard four-dimensional interconnection obtained by coupling an AIC to a
two-dimensional linear plant, the closed loop has a unique equilibrium that
may be unstable for some parameter values. General results for strongly
$2$-cooperative systems show that bounded trajectories whose omega-limit sets
avoid the equilibria exhibit a Poincar\'e--Bendixson-type behavior and converge
to periodic orbits~\cite{Eyal_k_posi,katz2025instability}. The corresponding
structure of compact omega-limit sets for the antithetic feedback system was
studied specifically in
\cite{2019biorxiv_margaliot_sontag,Margaliot-CDC}. These results leave
boundedness as an essential hypothesis: before the asymptotic structure can be
used, one must first exclude trajectories that escape to infinity.

This boundedness question was addressed in
\cite{wafi-boundedness-antithetic}, where all trajectories of the antithetic
feedback system considered there were shown to be bounded. The proof did not
rely on a Lyapunov function, and instead used a time-domain argument tailored
to the antithetic interconnection: if the controller output remains large for
long enough, the plant eventually generates a sufficiently large feedback
signal, which in turn forces the controller output to decrease. This suggests
that the relevant mechanism is not exclusively tied to the algebraic form of
the antithetic reactions. Rather, it depends on two complementary qualitative
input--output behaviors of the plant and controller.

Motivated by this observation, we introduce two properties for positive
nonlinear systems. The first is a
\emph{persistent-input/persistent-output} (PIPO) property for the plant. It
requires that, for every prescribed output level $\kappa$, an input that
remains above a suitable threshold $U(\kappa)$ eventually forces the plant
output to remain above $\kappa$, uniformly over the initial plant state. The
second is a \emph{persistent-input/transient-output} (PITO) property for the
controller. It requires that, for every tolerance $\varepsilon>0$, a
controller input that remains above a suitable threshold $V(\varepsilon)$
eventually forces the controller output below $\varepsilon$. The corresponding
transient time is required to depend on the initial controller state only
through its initial output, uniformly over internal states having the same
output.

These properties differ in direction from familiar bounded-input or
upper-gain stability notions. PIPO provides a lower bound on the eventual
output produced by a persistently large input, whereas PITO describes
attenuation of the output under a persistent lower bound on the input. Neither
property, by itself, asserts stability of the subsystem or boundedness under
arbitrary inputs. Their role is instead complementary: after the two systems
are interconnected, the plant converts a persistently large control signal
into a persistently large measurement, while the controller converts that
measurement into a sufficiently small control signal. The resulting
contradiction prevents the control signal from remaining arbitrarily large.

Our main result formalizes this mechanism. We consider a feedback
interconnection of positive nonlinear SISO systems in which the plant input is
the controller output and the controller input is the plant output. Assuming
that the closed-loop interconnection is forward complete and that the
controller output satisfies a one-sided affine growth bound, we prove that a
PIPO plant interconnected with a PITO controller generates a bounded control
signal (Theorem~\ref{thm:no_windup}). The theorem does not require linearity or
monotonicity of either subsystem once the PIPO and PITO properties have been
established. It should, however, be interpreted as a controller-output
boundedness result: without additional assumptions, it does not guarantee
boundedness of the complete controller state. If the plant is also
bounded-input bounded-state stable, boundedness of the plant state follows as
a direct consequence (Corollary~\ref{cor:PIPOTO_BIBS}).

We next provide conditions and examples showing how the two properties can be
verified. On the plant side, we prove that a positive monotone system with a
globally asymptotically stable equilibrium for every constant input and a
class-$\mathcal K_\infty$ input--output characteristic is PIPO and BIBS
stable~\cite{mcs_angeli_2003,AngeliS2004}. This gives a structural sufficient
condition for PIPO; an explicit transient-time bound additionally requires
quantitative information about convergence to the input-state
characteristic. On the controller side, we establish positivity and forward
completeness of the AIC and derive explicit PITO thresholds and transient
times. The main theorem then guarantees boundedness of the AIC output when it
is connected to any positive PIPO plant, provided that the closed loop is
forward complete. Boundedness of the second antithetic species does not follow
from the general input--output argument and may require additional
model-dependent structure.

Finally, we revisit the ``nonlinear II'' integral-feedback motif studied
in~\cite{shoval2011}. For a fixed positive external stimulus, we decompose the
model into a scalar controller that becomes an integral controller in
logarithmic coordinates and a stable linear plant. We verify PITO for the
controller and PIPO for the plant using explicit formulas. This example
illustrates the branch of the main theorem in which the controller output has
a multiplicative one-sided growth bound, complementing the additive growth
bound furnished by the antithetic controller.

The remainder of the paper is organized as follows.
Section~\ref{sec:prelim_def} introduces the positive plant--controller
interconnection and the standing assumptions.
Section~\ref{sec:bndnss_cs} defines the PIPO and PITO properties and proves the
bounded-control-signal theorem.
Section~\ref{sec:exmpl} develops the monotone-plant condition and analyzes the
antithetic controller and nonlinear-II motif.
Section~\ref{sec:conclusion} summarizes the results and discusses the
limitations of the framework.

\section{Preliminary Definitions}
\label{sec:prelim_def}

We use \(\Rp:=[0,\infty)\) and the standard partial order on \(\mathbb R^n\)
induced by the nonnegative orthant. Thus, for \(x,\tilde x\in\mathbb R^n\), we
write \(x\le \tilde x\) if \(\tilde x-x\in\mathbb R^n_{\ge0}\). Throughout the
paper, \(\|\cdot\|\) denotes the Euclidean norm.

A system
\[
    \dot x=f(x,u),\qquad y=h(x),
\]
with state \(x\in\mathbb R^n_{\ge0}\), input \(u\in\mathbb R_{\ge0}\), and
output \(y\in\mathbb R_{\ge0}\), is said to be \emph{positive} if, for every
locally essentially bounded input \(u:\mathbb R_{\ge0}\to\mathbb R_{\ge0}\)
and every initial condition \(x_0\in\mathbb R^n_{\ge0}\), the corresponding
solution satisfies
\[
    x(t)\in\mathbb R^n_{\ge0},
    \qquad
    y(t)=h(x(t))\in\mathbb R_{\ge0}
\]
for all times in its interval of existence.

The system is said to be \emph{forward complete} if, for every locally
essentially bounded input \(u:\mathbb R_{\ge0}\to\mathbb R_{\ge0}\) and every
initial condition \(x_0\in\mathbb R^n_{\ge0}\), the corresponding solution
exists for all \(t\ge0\).

The system is said to be \emph{bounded-input bounded-state (BIBS) stable} if,
for every initial condition \(x_0\in\mathbb R^n_{\ge0}\) and every constant
\(M_u>0\), there exists \(M_x=M_x(x_0,M_u)>0\) such that, for every locally
essentially bounded input \(u:\mathbb R_{\ge0}\to\mathbb R_{\ge0}\) satisfying
\[
    0\le u(t)\le M_u
    \qquad
    \forall t\ge0,
\]
the corresponding solution satisfies
\[
    \|x(t)\|\le M_x
    \qquad
    \forall t\ge0.
\]

Finally, a positive system is said to be \emph{monotone} if the output map
\(h\) is order preserving and, whenever
\[
    x_0\le \tilde x_0,
    \qquad
    u(t)\le \tilde u(t)
    \quad \forall t\ge0,
\]
one has
\[
    x(t,x_0,u)\le x(t,\tilde x_0,\tilde u)
    \qquad
    \forall t\ge0.
\]
Monotonicity will not be assumed in the general results; it will only be used
later as a sufficient condition for the PIPO property.

\subsection{Problem Setup}

\begin{figure}[t]
\centering
\begin{tikzpicture}[
    auto,
    >=Latex,
    block/.style={
        draw,
        rectangle,
        rounded corners,
        minimum height=1.1cm,
        minimum width=2cm,
        align=center
    },
    controllerblock/.style={
        draw,
        rectangle,
        rounded corners,
        minimum height=2.4cm,
        minimum width=4.8cm,
        align=center
    },
    sum/.style={
        draw,
        circle,
        inner sep=1pt,
        minimum size=5mm
    },
    dashedblock/.style={
        draw,
        dashed,
        rounded corners,
        inner sep=12pt
    }
]

\node[block] (controller) at (0,0)
{
$C$\\[1mm]
$\dot z=g(z,v)$\\
$w=k(z)$
};

\node[block] (plant) at ($(controller.north)+(0,1.8)$)
{
$P$\\[1mm]
$\dot x=f(x,u)$\\
$y=h(x)$
};

\node (yout) at ($(plant.east)+(2.0,0)$) {};
\node (wout) at ($(controller.west)-(2.0,0)$) {};
\coordinate (ysplit) at ($(plant.east)+(1.1,0)$);
\coordinate (wsplit) at ($(controller.west)-(1.1,0)$);

\draw[-] (controller.west) -- (wsplit);
\fill (wsplit) circle (1.2pt);
\draw[->] (wsplit) -- node[above] {$w$} (wout);
\draw[->] (wsplit) |- node[above right] {$u$} (plant.west);

\draw[-] (plant.east) -- (ysplit);
\fill (ysplit) circle (1.2pt);
\draw[->] (ysplit) -- node[above] {$y$} (yout);

\draw[->] (ysplit) |- node[above left] {$v$} (controller.east);

\end{tikzpicture}

\caption{Closed-loop interconnection between the plant \(P\) and the controller \(C\), with \(u=w\) and \(v=y\).}
\label{fig:plant-controller-interconnection}
\end{figure}
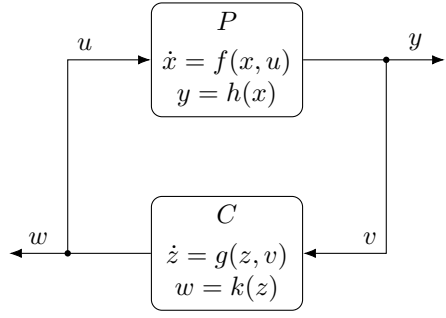

We study the interconnection of two positive control systems
\begin{subequations}
\label{eq:system-def}
\begin{align}
P:\quad
&\begin{cases}
    \dot x=f(x,u),\\
    y=h(x),\\
    x(0)=x_0\in\mathbb R^n_{\ge0},
\end{cases}
\label{eq:x-subsystem}\\
C:\quad
&\begin{cases}
    \dot z=g(z,v),\\
    w=k(z),\\
    z(0)=z_0\in\mathbb R^m_{\ge0},
\end{cases}
\label{eq:z-subsystem}
\end{align}
\end{subequations}
where
\[
    f:\mathbb R^n_{\ge0}\times\mathbb R_{\ge0}\to\mathbb R^n,
    \qquad
    g:\mathbb R^m_{\ge0}\times\mathbb R_{\ge0}\to\mathbb R^m,
\]
and
\[
    h:\mathbb R^n_{\ge0}\to\mathbb R_{\ge0},
    \qquad
    k:\mathbb R^m_{\ge0}\to\mathbb R_{\ge0}
\]
are locally Lipschitz. The interconnection is
\[
    u=w=k(z),
    \qquad
    v=y=h(x),
\]
as illustrated in Figure~\ref{fig:plant-controller-interconnection}.

We assume throughout the main result that each system is forward-complete under locally essentially bounded positive inputs, and that the closed-loop interconnection is
forward complete. We also assume a one-sided affine growth bound for the controller output, that is, there exists
\(\omega,\overline k\ge0\) such that
\[
    \dot w(t)
    =
    \frac{d}{dt}k(z(t))
    \le \overline k w(t) +\omega
    \qquad
    \forall t\ge0.
\]
This one-sided growth assumption is the natural condition in the positive
setting. 

We will next look at two conditions that, if satisfied by the plant and controller, guarantee boundedness of the control signal under interconnection.

\section{On the boundedness of the control signal}
\label{sec:bndnss_cs}

In this section we will study specific properties that, if satisfied by the plant and the controller, will guarantee boundedness of the control signal. Informally, for the plant $P$ \eqref{eq:x-subsystem}, we require that it persistently amplifies its input signal, that is that for every $\kappa>0$ there exists some threshold $U_\kappa>0$ such that 
\[
u(t)\ge U_\kappa \quad \forall t\ge0 \quad\implies\quad \liminf_{t\to\infty}h(x(t,x_0))\ge \kappa,
\]
``uniformly'' for every initial condition $x_0$.
See Figure~\ref{fig:PIPO_diagram} for an illustration of this property.
\begin{figure}[ht]
  \begin{center}
\begin{tikzpicture}[>=Stealth, thick]

\begin{scope}
  \draw[->] (0,0) -- (0,2.6) node[above] {$u$};
  \draw[->] (0,0) -- (3.2,0) node[right] {$t$};

  \draw[dashed, green!60!black] (0,1.4) -- (3.0,1.4);
  
  \node at (-0.4,1.4) {$U(\kappa)$};

  \draw[thick, blue] (0,1.9)
    .. controls (0.5,2.3) and (0.9,1.5) .. (1.4,1.9)
    .. controls (1.9,2.3) and (2.3,1.5) .. (2.8,1.9);
\end{scope}

\begin{scope}[xshift=4.7cm]
  \draw (0,0) rectangle (2.6,2.2);
  \node at (1.3,1.1) {\Large PIPO};

  \draw[->] (-1.0,1.1) -- (0,1.1);

  \draw[->] (2.6,1.1) -- (3.6,1.1);
\end{scope}

\begin{scope}[xshift=9.2cm]
  \draw[->] (0,0) -- (0,2.6) node[above] {$y$};
  \draw[->] (0,0) -- (3.2,0) node[right] {$t$};

  \draw[dashed, green!60!black] (0,1.7) -- (3.0,1.7);
  \node at (-0.15,1.7) {$\kappa$};

  \draw[dashed, green!60!black] (2.23,0) -- (2.23,2.5);
  \node at (2.23,-0.25) {$T(\kappa)$};

  \draw[thick, blue]
    (0,2.4)
    .. controls (0.15,1.4) and (0.35,0.3) .. (0.6,0.4)
    .. controls (0.9,0.5) and (1.0,1.1) .. (1.3,1.0)
    .. controls (1.7,0.9) and (1.9,1.6) .. (2.2,1.7)
    .. controls (2.5,1.75) and (2.8,1.7) .. (3.0,1.7);
\end{scope}
\end{tikzpicture}
\end{center}
\caption{The PIPO property -- persistently large inputs eventually drive the system's output to also be proportionally persistently large.}
\label{fig:PIPO_diagram}
\end{figure}
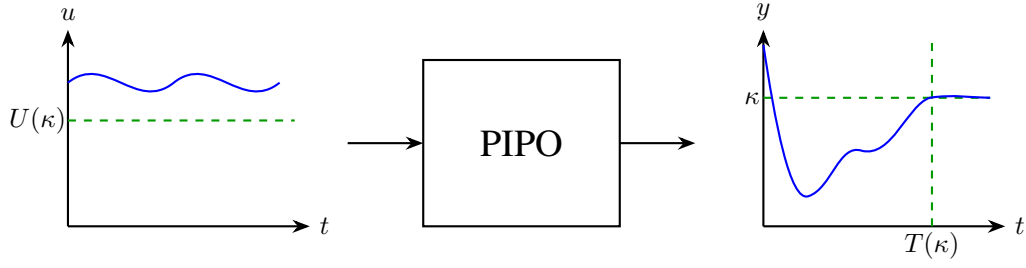

More precisely, consider the following definition.

\begin{definition}[persistent-input/persistent-output property]\label{def:PIPO}
    A positive system of the form of \eqref{eq:x-subsystem} is said to satisfy the persistent-input/persistent-output (PIPO) property if there exist functions 
    \[
    U:\Rp\to\Rp \text{\quad and \quad} T:\Rp\to\Rp 
    \] 
    such that for every $\kappa\ge0$, every $x_0\in\Rp[n]$ it holds that:
    \[
    u(t)\ge U(\kappa) \quad\forall \, t\ge0 \quad\implies \quad
    y(t)\geq\, \kappa \quad\forall\,t\ge T(\kappa).
    \]
\end{definition}

A closely related property is also required of the controller $C$ \eqref{eq:z-subsystem}. Specifically, while the plant is required to amplify its input signal, the controller must attenuate it proportionally to its magnitude, so that the resulting control signal after interconnection is bounded. Intuitively we ask that for every $\eps>0$ there exist some $V_\eps>0$ such that
\[
v(t)\ge V_\eps\quad\forall t\ge0\quad\implies\quad\limsup_{t\to\infty}k(z(t,z_0))\le\eps,
\]
``uniformly'', and for every initial condition $z_0$.
See Figure~\ref{fig:PITO_diagram} for an illustration of this property.
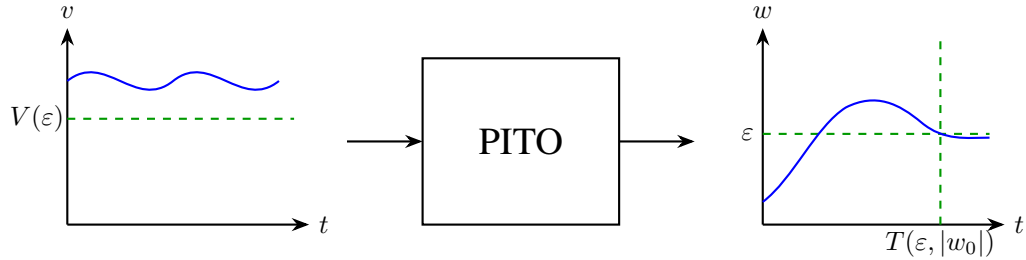
\begin{figure}[ht]
  \begin{center}
\begin{tikzpicture}[>=Stealth, thick]
\begin{scope}
  \draw[->] (0,0) -- (0,2.6) node[above] {$v$};
  \draw[->] (0,0) -- (3.2,0) node[right] {$t$};
  \draw[dashed, green!60!black] (0,1.4) -- (3.0,1.4);
  \node at (-0.4,1.4) {$V(\eps)$};
  \draw[thick, blue] (0,1.9)
    .. controls (0.5,2.3) and (0.9,1.5) .. (1.4,1.9)
    .. controls (1.9,2.3) and (2.3,1.5) .. (2.8,1.9);
\end{scope}
\begin{scope}[xshift=4.7cm]
  \draw (0,0) rectangle (2.6,2.2);
  \node at (1.3,1.1) {\Large PITO};
  \draw[->] (-1.0,1.1) -- (0,1.1);
  \draw[->] (2.6,1.1) -- (3.6,1.1);
\end{scope}
\begin{scope}[xshift=9.2cm]
  \draw[->] (0,0) -- (0,2.6) node[above] {$w$};
  \draw[->] (0,0) -- (3.2,0) node[right] {$t$};
  \draw[dashed, green!60!black] (0,1.2) -- (3.0,1.2);
  \node at (-0.2,1.2) {$\eps$};
  \draw[dashed, green!60!black] (2.35,0) -- (2.35,2.5);
  \node at (2.35,-0.25) {$T(\eps,|w_0|)$};
  \draw[thick, blue]
    (0,0.3)
    .. controls (0.4,0.6) and (0.7,1.3) .. (1.1,1.55)
    .. controls (1.5,1.75) and (1.8,1.6) .. (2.1,1.35)
    .. controls (2.4,1.1) and (2.7,1.15) .. (3.0,1.15);
\end{scope}
\end{tikzpicture}
\end{center}
\caption{The PITO property -- persistently large inputs eventually drive the system's output to be proportionally persistently small}
\label{fig:PITO_diagram}
\end{figure}  

More precisely, consider the following definition.

\begin{definition}[persistent-input/transient-output property]
\label{def:PITO}
A system of the form \eqref{eq:z-subsystem} is said to satisfy the persistent-input/transient-output (PITO)
property from $v$ to $w$ if
there exist two functions
\[
V : \Rp \to \Rp
\text{\quad and \quad}
T : \Rp[2] \to \Rp
\]
with \(T\) non-decreasing in its second argument, such that, for every $\eps>0$, and every 
$z(0)\in\Rp[m]$,
\[
v(t)\geq V(\varepsilon)
    \;\; \forall\,t\ge0
    \quad\implies\quad
    w(t)\leq \varepsilon\;\;\forall\,t\ge T(\eps,k(z(0))) .
\]
\end{definition}

The properties characterized in Definitions~\ref{def:PIPO} and~\ref{def:PITO} are complementary in nature, both characterizing a specific input-to-output behavior of the system -- roughly either amplification or attenuation of the input signal. 
We use these properties to prove the following result.

\begin{theorem}
    \label{thm:no_windup}
    Consider the positive systems \eqref{eq:x-subsystem} and \eqref{eq:z-subsystem}, interconnected by $u=w=k(z)$ and $v=y(t)=h(x(t))$. Assume that
    \begin{itemize}
        \item The defined closed-loop interconnection of both subsystems is forward-complete;
        \item The controller satisfies an affine output growth bound, i.e. there exist $\overline k,\omega>0$ such that $\dot w(t) \le \overline kw(t)+\omega$ hold for all $z\in\Rp[m]$; 
        \item The controller \eqref{eq:z-subsystem} satisfies the PITO property; and
        \item The plant \eqref{eq:x-subsystem} satisfies the PIPO property.
    \end{itemize}
    Then, for any initialization $z(0)\in\Rp[m]$ and $x(0)\in\Rp[n]$, the control signal $u(t)=w(t)=k(z(t))$ is bounded for all $t\ge 0$.
\end{theorem}

\begin{proof}
Fix an arbitrary initial condition
\[
    x(0)\in\mathbb R^n_{\ge0},
    \qquad
    z(0)\in\mathbb R^m_{\ge0}.
\]
Since the closed-loop interconnection is forward complete, the corresponding
solution exists for all \(t\ge0\). One can verify that the interconnection is positive as well, thus all
signals satisfy
\[
    u(t)\ge0,
    \qquad
    v(t)=y(t)\ge0,
    \qquad
    w(t)\ge0
    \qquad
    \forall t\ge0.
\]
By the one-sided affine growth assumption on the controller output,
\[
    \dot w(t)\le \overline k w(t)+\omega
    \qquad
    \forall t\ge0,
\]
and by Gronwall's inequality,
for every \(t_2\ge t_1\ge0\),
\[
    w(t_2)\le\phi\big(w(t_1),t_2-t_1\big),
\]
where
\[
    \phi(w_1,\Delta)
    :=
    \begin{cases}
        \left(w_1+\dfrac{\omega}{\overline k}\right)e^{\overline k\Delta}-\dfrac{\omega}{\overline k}, & \overline k>0,\\[2mm]
        w_1+\omega\Delta, & \overline k=0.
    \end{cases}
\]
For fixed \(w_1\ge0\), \(\phi(w_1,\cdot)\) is nondecreasing on
\(\mathbb R_{\ge0}\) and \(\phi(w_1,0)=w_1\).

Fix \(\varepsilon>0\), and let $V_\eps$ be the corresponding gain for the PITO property and system \eqref{eq:z-subsystem} for the fixed $\eps$. Furthermore, let $U_\eps$ be the corresponding gain for the PIPO property and system \eqref{eq:x-subsystem} for $\kappa = V_\eps$, i.e. $U_\eps := U_{V_\eps}$.
Thus, for every initial condition \(x_0\in\mathbb R^n_{\ge0}\),
\[
    u(t)\ge U_\varepsilon
    \qquad
    \forall t\ge0
\]
implies
\[
    y(t)\ge V_\varepsilon
    \qquad
    \forall t\ge T_\varepsilon^x,
\]
where $T_\eps^x$ is the PIPO-associated time for $\kappa=V_\eps$.
Since the plant is time invariant and the PIPO property holds for every
initial condition in \(\mathbb R^n_{\ge0}\), the same implication holds on
shifted intervals. Namely, if
\[
    u(t)\ge U_\varepsilon
    \qquad
    \forall t\in[a,b],
    \qquad
    b-a\ge T_\varepsilon^x,
\]
then
\[
    y(t)\ge V_\varepsilon
    \qquad
    \forall t\in[a+T_\varepsilon^x,b].
\]
Choose any \(\eta>0\), and define
\[
    A:=\max\{U_\varepsilon,w(0),\varepsilon+\eta\}.
\]
Set
\[
    \delta:=\phi(A,T_\varepsilon^x).
\]
Let
\[
    T_\varepsilon^z:=T(\varepsilon,\delta)
\]
be the PITO time corresponding to accuracy \(\varepsilon\) and initial
controller-output bound \(\delta\). Finally, define
\[
    R
    :=
    \phi\big(A,\,T_\varepsilon^x+T_\varepsilon^z\big)+\eta.
\]
Since \(\phi(A,\cdot)\) is nondecreasing with \(\phi(A,0)=A\), we have
\(\phi(A,T_\varepsilon^x+T_\varepsilon^z)\ge A\), so the extra \(+\eta\)
ensures \(R\ge A+\eta>A\).
We claim that
\[
    w(t)\le R
    \qquad
    \forall t\ge0.
\]
Suppose, by contradiction, that this is false. Since \(w(0)\le A<R\), there
exists a first time \(\tau>0\) such that
\[
    w(\tau)=R,
    \qquad
    w(t)<R
    \quad
    \forall t\in[0,\tau).
\]
Since \(R>A\ge w(0)\), there exists a last time
\(\sigma\in[0,\tau)\) such that
\[
    w(\sigma)=A.
\]
By the definition of \(\sigma\),
\[
    w(t)\ge A
    \qquad
    \forall t\in[\sigma,\tau].
\]
Since \(u=w\) and \(A\ge U_\varepsilon\), it follows that
\[
    u(t)\ge U_\varepsilon
    \qquad
    \forall t\in[\sigma,\tau].
\]
Using the one-sided growth estimate, we obtain
\[
    R
    =
    w(\tau)
    \le
    \phi(w(\sigma),\tau-\sigma)
    =
    \phi(A,\tau-\sigma).
\]
By the definition of \(R\), this gives
\[
    \phi(A,T_\varepsilon^x+T_\varepsilon^z)+\eta
    \le
    \phi(A,\tau-\sigma),
\]
which forces \(\tau-\sigma\ge T_\varepsilon^x+T_\varepsilon^z\), since \(\phi(A,\cdot)\) is nondecreasing.

Define
\[
    s:=\sigma+T_\varepsilon^x.
\]
Since
\[
    \tau-\sigma\ge T_\varepsilon^x+T_\varepsilon^z,
\]
we have
\[
    s+T_\varepsilon^z\le \tau.
\]
By the shifted PIPO implication,
\[
    y(t)\ge V_\varepsilon
    \qquad
    \forall t\in[s,s+T_\varepsilon^z].
\]
Since \(v=y\) in the closed-loop interconnection, this gives
\[
    v(t)\ge V_\varepsilon
    \qquad
    \forall t\in[s,s+T_\varepsilon^z].
\]
Next, we show that the controller output at time \(s\) is bounded by
\(\delta\). Since \(s=\sigma+T_\varepsilon^x\), the one-sided growth estimate
gives
\[
    w(s)
    \le
    \phi(w(\sigma),T_\varepsilon^x)
    =
    \phi(A,T_\varepsilon^x)
    =
    \delta.
\]
Since \(w=k(z)\), this means
\[
    k(z(s))\le \delta.
\]
We now apply the PITO property to the controller trajectory shifted to time
\(s\). Because the PITO \(T\) is nondecreasing in its second argument and
\[
    k(z(s))\le \delta,
\]
we have
\[
    T(\varepsilon,k(z(s)))
    \le
    T(\varepsilon,\delta)
    =
    T_\varepsilon^z.
\]
Moreover,
\[
    v(t)\ge V_\varepsilon
    \qquad
    \forall t\in[s,s+T_\varepsilon^z].
\]
By causality, the value of \(w(s+T_\varepsilon^z)\) depends only on the input
over the interval \([s,s+T_\varepsilon^z]\). Therefore the PITO implication
applied to the shifted controller solution gives
\[
    w(s+T_\varepsilon^z)\le \varepsilon.
\]
On the other hand,
\[
    s+T_\varepsilon^z\in[\sigma,\tau],
\]
and by construction
\[
    w(t)\ge A
    \qquad
    \forall t\in[\sigma,\tau].
\]
Therefore
\[
    w(s+T_\varepsilon^z)\ge A>\varepsilon,
\]
which contradicts the PITO conclusion that
\[
    w(s+T_\varepsilon^z)\le \varepsilon.
\]
Hence, by contradiction
\[
    w(t)\le R
    \qquad
    \forall t\ge0.
\]
Since \(u=w=k(z)\), the control signal \(u(t)\) is bounded on
\([0,\infty)\).
\end{proof}

Intuitively, one can understand the interconnection of a PIPO and a PITO system as a high-gain negative feedback for the control input: the PIPO system amplifies the control signal, which is then attenuated by the PITO system. This intuitive idea behind the proof of Theorem~\ref{thm:no_windup} is illustrated in Figure~\ref{fig:PITO_PITO}.

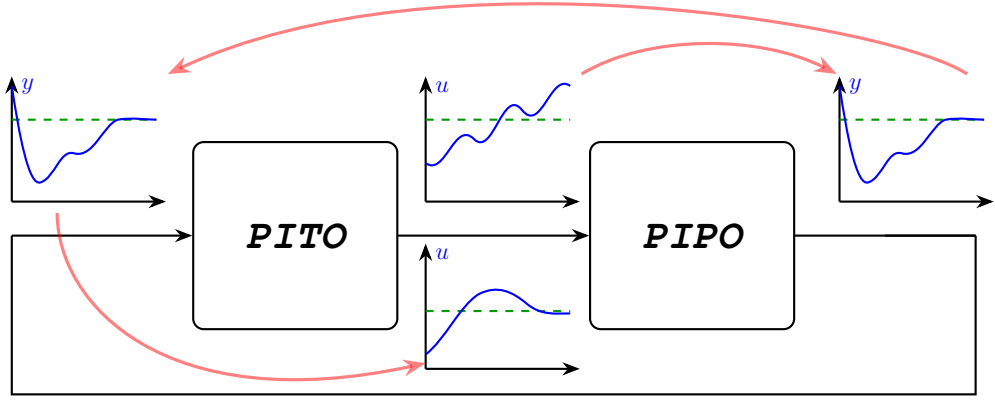
\begin{figure}[ht]
\begin{center}
\begin{tikzpicture}[>=Stealth, thick, scale=0.75, transform shape]

\node[draw, rounded corners, minimum width=3.6cm, minimum height=3.3cm,
      font=\ttfamily\bfseries\itshape\huge] (pito) at (6,0)  {PITO};
\node[draw, rounded corners, minimum width=3.6cm, minimum height=3.3cm,
      font=\ttfamily\bfseries\itshape\huge] (piso) at (13,0) {PIPO};

\coordinate (splitL) at (1.0,0);
\coordinate (splitR) at (18.0,0);
\coordinate (outEdge) at (16.4,0);

\draw[->] (splitL) -- (pito.west);
\draw[->]  (pito.east) -- (piso.west);
\draw[-]  (piso.east) -- (splitR);
\draw[-] (splitR) -- (outEdge);

\draw[-] (splitR) |- ++(0,-2.8) -| (splitL);

\begin{scope}[shift={(1.0,0.6)}, scale=0.85]
  \draw[->] (0,0) -- (0,2.6) node[right, blue, xshift=2pt, pos=0.92] {\Large $y$};
  \draw[->] (0,0) -- (3.2,0);
  \draw[dashed, green!60!black] (0,1.7) -- (3.0,1.7);
  \draw[thick, blue]
    (0,2.4)
    .. controls (0.15,1.4) and (0.35,0.3) .. (0.6,0.4)
    .. controls (0.9,0.5) and (1.0,1.1) .. (1.3,1.0)
    .. controls (1.7,0.9) and (1.9,1.6) .. (2.2,1.7)
    .. controls (2.5,1.75) and (2.8,1.7) .. (3.0,1.7);
\end{scope}

\begin{scope}[shift={(15.6,0.6)}, scale=0.85]
  \draw[->] (0,0) -- (0,2.6) node[right, blue, xshift=2pt, pos=0.92] {\Large $y$};
  \draw[->] (0,0) -- (3.2,0);
  \draw[dashed, green!60!black] (0,1.7) -- (3.0,1.7);
  \draw[thick, blue]
    (0,2.4)
    .. controls (0.15,1.4) and (0.35,0.3) .. (0.6,0.4)
    .. controls (0.9,0.5) and (1.0,1.1) .. (1.3,1.0)
    .. controls (1.7,0.9) and (1.9,1.6) .. (2.2,1.7)
    .. controls (2.5,1.75) and (2.8,1.7) .. (3.0,1.7);
\end{scope}

\begin{scope}[shift={(8.3,0.6)}, scale=0.85]
  \draw[->] (0,0) -- (0,2.6) node[right, blue, xshift=2pt, pos=0.92] {\Large $u$};
  \draw[->] (0,0) -- (3.2,0);
  \draw[dashed, green!60!black] (0,1.7) -- (3.0,1.7);
  \draw[thick, blue]
    (0,0.8)
    .. controls (0.4,0.5) and (0.6,1.7) .. (1.0,1.3)
    .. controls (1.4,0.9) and (1.6,2.4) .. (2.0,1.9)
    .. controls (2.4,1.4) and (2.6,2.7) .. (3.0,2.4);
\end{scope}


\begin{scope}[shift={(8.3,-2.35)}, scale=0.85]
  \draw[->] (0,0) -- (0,2.6) node[right, blue, xshift=2pt, pos=0.92] {\Large $u$};
  \draw[->] (0,0) -- (3.2,0);
  \draw[dashed, green!60!black] (0,1.2) -- (3.0,1.2);
  \draw[thick, blue]
    (0,0.3)
    .. controls (0.4,0.6) and (0.7,1.3) .. (1.1,1.55)
    .. controls (1.5,1.75) and (1.8,1.6) .. (2.1,1.35)
    .. controls (2.4,1.1) and (2.7,1.15) .. (3.0,1.15);
\end{scope}


\coordinate (midTopTR) at (11.05,2.85);  
\coordinate (rightTL)  at (15.55,2.85);  

\coordinate (rightTR) at (17.85, 2.85);

\coordinate (rightBL)  at (15.55,0.55);  
\coordinate (leftBR)   at (3.75,0.55);   

\coordinate (leftTR)   at (3.75,2.85);   

\coordinate (leftR)    at (3.75,1.70);   
\coordinate (leftB)    at (1.8,0.4);
\coordinate (midBotL)  at (8.30,-2.25);  

\draw[->, red, line width=1.2pt, opacity=0.5,]
    (midTopTR)
    .. controls (12.2,3.55) and (14.2,3.55) ..
    (rightTL);

\draw[->, red, line width=1.2pt, opacity=0.5,]
    (rightTR)
    .. controls (16.6,3.7) and (8.5,5.1) ..
    (leftTR);

\draw[->, red, line width=1.2pt, opacity=0.5,]
    (leftB)
    .. controls (1.8,-1.2) and (3.7,-3.25) ..
    (midBotL);
\end{tikzpicture}
\end{center}
\caption{Intuition on PIPO/PITO interconnections.
The input $u$ (top middle signal) being large implies that the PIPO
plant produces an eventually large $y$,
which as input to a PITO controller (after a long enough time) implies
boundedness of $u$.}
\label{fig:PITO_PITO}
\end{figure} 

The intuitive mechanics behind the result is that the plant being PIPO means it is guaranteed to amplify any signal in its input, while the controller being PITO means that any persistent input will attenuate its output, which finally guarantees boundedness of the control signal. Notice, however, that Theorem~\ref{thm:no_windup} does not guarantee by itself boundedness of all relevant states or even of the output of the plant. From Theorem~\ref{thm:no_windup}, boundedness of all states can follow from bounded-input bounded-state (BIBS) stability through the following corollary. 
\begin{corollary}
    \label{cor:PIPOTO_BIBS}
    Assume the \(x\)-subsystem \eqref{eq:x-subsystem} and the \(z\)-subsystem \eqref{eq:z-subsystem} satisfy the conditions in Theorem~\ref{thm:no_windup}. Furthermore, assume the \(x\)-subsystem is BIBS stable with respect to the input \(u\). Then, for any initial condition \((x_0,z_0)\), the solution of the \(x\)-subsystem under the closed-loop interconnection, \(\xsol(t)\), is bounded on \([0,\infty)\).
\end{corollary}

\begin{proof}
Fix an arbitrary initial condition
\[
    (x(0),z(0))=(x_0,z_0)\in\mathbb R^n\times\mathbb R^m .
\]
By Theorem~\ref{thm:no_windup}, the closed-loop control signal
\[
    u(t)=w(t)=k(z(t))
\]
is bounded on \([0,\infty)\). Therefore, there exists
\(M_u>0\) such that
\[
    u(t)\le M_u
    \qquad
    \forall t\ge0.
\]
Since the \(x\)-subsystem is BIBS stable with respect to the input \(u\), it
follows that there exists \(M_x=M_x(x_0,M_u)>0\) such that
\[
    \|\xsol(t)\|\le M_x
    \qquad
    \forall t\ge0.
\]
Hence the \(x\)-subsystem trajectory under the closed-loop interconnection is
bounded on \([0,\infty)\).
\end{proof}

\begin{remark}
    Notice that BIBS stability is, in general, too strong of an assumption to make of a controller, and if one wants to prove boundedness of all states under interconnection, a more adhoc analysis is required. If, however, one can certify BIBS stability of both plant and controller, boundedness of all states follows naturally by extending the previous Corollary.
\end{remark}

\section{Sufficient conditions for the PIPO and PITO properties}
\label{sec:exmpl}

Having established the key PIPO and PITO properties for the plant and the controller, we next focus on the class of plants that fit this framework. Specifically, we will look at a class of monotone systems, showing that, under very mild conditions on the input-output characteristic, one can show that PIPO holds.

\subsection{PIPO and monotone systems}

Consider the following assumptions.

\begin{assumption}[Monotonicity \cite{mcs_angeli_2003,AngeliS2004}]
\label{ass:monotone}
The plant subsystem~\eqref{eq:x-subsystem} is positive and monotone with
respect to both state and input, and the output map \(h\) is order preserving.
That is, for every nonnegative input
\(u:\mathbb R_{\ge0}\to\mathbb R_{\ge0}\) and every initial condition
\(x_0\in\mathbb R^n_{\ge0}\),
\[
    \xsol(t,x_0,u)\in\mathbb R^n_{\ge0}
    \qquad
    \forall t\ge0.
\]
Moreover, if \(x_0\le \tilde x_0\) and
\(u(t)\le \tilde u(t)\) for all \(t\ge0\), then
\[
    \xsol(t,x_0,u)\le \xsol(t,\tilde x_0,\tilde u)
    \qquad
    \forall t\ge0,
\]
and, whenever \(x\le \tilde x\),
\[
    h(x)\le h(\tilde x).
\]
\end{assumption}

\begin{assumption}[Class $\mathcal{K}_\infty$ characteristic]
    \label{ass:characteristic}
    The plant subsystem \eqref{eq:x-subsystem} admits a well-defined input-state characteristic, given by
\[
    k_x:\mathbb R_{\ge0}\to\mathbb R^n_{\ge0},
    \qquad
    \bar u\mapsto k_x(\bar u),
\]
where \(k_x(\bar u)\) is the globally asymptotically stable equilibrium corresponding to the constant input
\(\bar u\). The associated input-output characteristic is
\[
    k_y:\mathbb R_{\ge0}\to\mathbb R_{\ge0},
    \qquad
    k_y(\bar u):=h(k_x(\bar u)).
\]
and is of class \(\mathcal K_\infty\);
that is, \(k_y\) is continuous, strictly increasing,
\[
    k_y(0)=0,
\]
and
\[
    \lim_{\bar u\to\infty}k_y(\bar u)=\infty .
\]
\end{assumption}

Under these assumptions, stable monotone plants satisfy the PIPO property.

\begin{lemma}
\label{lem:monotone-pipo}
Suppose that the plant subsystem~\eqref{eq:x-subsystem} satisfies
Assumptions~\ref{ass:monotone}, and~\ref{ass:characteristic}.
Then the plant is BIBS stable and satisfies the PIPO property.
\end{lemma}

\begin{proof}
We prove the two claims separately.

\emph{First, we prove BIBS stability.} Fix an arbitrary initial condition
\(x_0\in\mathbb R^n_{\ge0}\) and an arbitrary constant \(M_u>0\). Let
\(u:\mathbb R_{\ge0}\to\mathbb R_{\ge0}\) be any admissible input satisfying
\[
    0\le u(t)\le M_u
    \qquad
    \forall t\ge0.
\]
By monotonicity with respect to the input, using the same initial condition
\(x_0\), we have
\[
    \xsol(t,x_0,u)\le \xsol(t,x_0,M_u)
    \qquad
    \forall t\ge0,
\]
where \(M_u\) denotes the constant input \(u(t)\equiv M_u\). By
Assumption~\ref{ass:characteristic},
\[
    \xsol(t,x_0,M_u)\to k_x(M_u)
    \qquad
    \text{as } t\to\infty .
\]
In particular, the trajectory \(t\mapsto \xsol(t,x_0,M_u)\) is bounded.
Therefore, there exists \(\overline M_x=\overline M_x(x_0,M_u)>0\) such that
\[
    \|\xsol(t,x_0,M_u)\|\le \overline M_x
    \qquad
    \forall t\ge0.
\]
Since the system is positive,
\[
    \xsol(t,x_0,u)\in\mathbb R^n_{\ge0}
    \qquad
    \forall t\ge0.
\]
Together with
\[
    0\le \xsol(t,x_0,u)\le \xsol(t,x_0,M_u),
\]
this implies
\[
    \|\xsol(t,x_0,u)\|\le \overline M_x
    \qquad
    \forall t\ge0.
\]
Hence the plant is BIBS stable with respect to nonnegative inputs.

\emph{We now prove the PIPO property.} The case \(\kappa=0\) is trivial, so fix
\(\kappa>0\). Since \(k_y\in\mathcal K_\infty\), there exists
\(U_\kappa>0\) such that
\[
    k_y(U_\kappa)>2\kappa .
\]
Consider the reference trajectory initialized at the origin and driven by the
constant input \(U_\kappa\):
\[
    x_r(t):=\xsol(t,0,U_\kappa),
    \qquad
    y_r(t):=h(x_r(t)).
\]
By Assumption~\ref{ass:characteristic},
\[
    x_r(t)\to k_x(U_\kappa)
    \qquad
    \text{as } t\to\infty .
\]
Since \(h\) is continuous, it follows that
\[
    y_r(t)=h(x_r(t))\to h(k_x(U_\kappa))
    =
    k_y(U_\kappa)
    >
    2\kappa .
\]
Therefore, there exists \(T_\kappa\ge0\) such that
\[
    y_r(t)\ge \kappa
    \qquad
    \forall t\ge T_\kappa .
\]

Now fix any \(\overline T\ge T_\kappa\), any initial condition
\(x_0\in\mathbb R^n_{\ge0}\), and any admissible nonnegative input
\(u:\mathbb R_{\ge0}\to\mathbb R_{\ge0}\) satisfying
\[
    u(t)\ge U_\kappa
    \qquad
    \forall t\ge0.
\]
Since \(0\le x_0\) and \(U_\kappa\le u(t)\) for all \(t\ge0\), monotonicity
gives
\[
    x_r(t)
    =
    \xsol(t,0,U_\kappa)
    \le
    \xsol(t,x_0,u)
    \qquad
    \forall t\ge0.
\]
Because \(h\) is order preserving,
\[
    y(t)
    =
    h(\xsol(t,x_0,u))
    \ge
    h(x_r(t))
    =
    y_r(t)
    \qquad
    \forall t\ge0.
\]
In particular, at \(t=\overline T\),
\[
    y(\overline T)\ge y_r(\overline T)\ge \kappa.
\]
Since the output is nonnegative, this implies
\[
    |y(\overline T)|\ge \kappa .
\]
Thus, the PIPO implication holds for every \(\overline T\ge T_\kappa\), and hence the
plant satisfies the PIPO property.
\end{proof}

\subsection{PITO and the Antithetic Integral Controller}

We next look at the antithetic integral controller \cite{AFC2016}, which is an important control structure in synthetic biology. We will first show that it  is positive and forward complete and then show that it satisfies the PITO property.

Consider the antithetic integral controller in its usual positive-systems
form
\begin{equation}
\label{eq:antithetic_controller}
C_{\mathrm{AIC}}:
\begin{cases}
    \dot z_1=\alpha_1-\alpha_2z_1z_2,\\[1mm]
    \dot z_2=\alpha_3v-\alpha_4z_1z_2,\\[1mm]
    w=z_1,
\end{cases}
\end{equation}
where \(z=(z_1,z_2)\in\mathbb R^2_{\ge0}\), \(y\in\mathbb R_{\ge0}\), and
\(\alpha_i>0\). In the standard antithetic feedback interconnection, the input
\(v\) is the nonnegative plant output \(y\), and the controller output \(w=z_1\) is
the plant input \(u\).

The controller \eqref{eq:antithetic_controller} is naturally defined on the
positive cone. We first record positive invariance and forward completeness, showing that the positive cone is the natural state space of \eqref{eq:antithetic_controller}.

\begin{lemma}[Positive invariance and forward completeness]
\label{lem:aic-positive-forward-complete}
For every locally essentially bounded input
\(v:\mathbb R_{\ge0}\to\mathbb R_{\ge0}\) and every initial condition
\(z(0)\in\mathbb R^2_{\ge0}\), the solution of
\eqref{eq:antithetic_controller} is forward complete and satisfies
\[
    z(t)\in\mathbb R^2_{\ge0}
    \qquad
    \forall t\ge0.
\]
\end{lemma}

\begin{proof}
Let
\[
    K:=\mathbb R^2_{\ge0}.
\]
We first prove forward invariance of \(K\). Since \(K\) is closed and convex,
Nagumo's tangent-cone criterion states that \(K\) is forward invariant if
\[
    F(z,v)\in T_K(z)
    \qquad
    \forall z\in K,\quad \forall v\ge0,
\]
where \(F\) is the vector field of \eqref{eq:antithetic_controller} and
\(T_K(z)\) is the tangent cone of \(K\) at \(z\).

For the positive orthant,
\[
    T_K(z)
    =
    \left\{
        q\in\mathbb R^2:
        q_i\ge0 \text{ whenever } z_i=0
    \right\}.
\]
Thus it is enough to check the two boundary faces. If \(z_1=0\), then
\[
    \dot z_1=\alpha_1>0.
\]
If \(z_2=0\), then
\[
    \dot z_2=\alpha_3v\ge0,
\]
because \(v\ge0\). Hence the vector field belongs to \(T_K(z)\) at every
boundary point \(z\in K\). By Nagumo's theorem, \(K=\mathbb R^2_{\ge0}\) is
forward invariant. Therefore, for every \(z(0)\in\mathbb R^2_{\ge0}\),
\[
    z(t)\in\mathbb R^2_{\ge0}
\]
for all \(t\) in the interval of existence.

We now prove forward completeness. By local Lipschitzness of the vector field
in \(z\), for every initial condition \(z(0)\in\mathbb R^2_{\ge0}\) there
exists a unique solution on a maximal interval of existence
\[
    [0,t_{\max}),
    \qquad
    0<t_{\max}\le\infty .
\]
We show that \(t_{\max}=\infty\).

Suppose, for contradiction, that \(t_{\max}<\infty\). Since the solution remains
in \(\mathbb R^2_{\ge0}\), we have \(z_1(t),z_2(t)\ge0\) for all
\(t\in[0,t_{\max})\). Hence
\[
    \dot z_1
    =
    \alpha_1-\alpha_2z_1z_2
    \le
    \alpha_1.
\]
Therefore,
\[
    0\le z_1(t)
    \le
    z_1(0)+\alpha_1 t
    \le
    z_1(0)+\alpha_1 t_{\max}
    \qquad
    \forall t\in[0,t_{\max}).
\]
Similarly,
\[
    \dot z_2
    =
    \alpha_3v-\alpha_4z_1z_2
    \le
    \alpha_3v.
\]
Since \(v\) is locally essentially bounded and \(t_{\max}<\infty\), there
exists \(M_v>0\) such that
\[
    0\le v(t)\le M_v
    \qquad
    \text{for a.e. } t\in[0,t_{\max}).
\]
Thus
\[
    0\le z_2(t)
    \le
    z_2(0)+\alpha_3M_v t
    \le
    z_2(0)+\alpha_3M_v t_{\max}
    \qquad
    \forall t\in[0,t_{\max}).
\]
Consequently, the solution satisfies
\[
    z(t)\in \mathcal K
    \qquad
    \forall t\in[0,t_{\max}),
\]
where
\[
    \mathcal K
    :=
    \left[0,z_1(0)+\alpha_1t_{\max}\right]
    \times
    \left[0,z_2(0)+\alpha_3M_v t_{\max}\right]
\]
is compact.

Since the vector field is locally Lipschitz in \(z\) and the input \(v\) is
locally essentially bounded on \([0,t_{\max})\), the standard continuation theorem implies that the solution can be extended beyond $t_{\max}$. Therefore $t_{\max}=\infty$; see, for instance, Proposition C.3.6 in \cite{mct}.
\end{proof}

We next show that the antithetic controller has the required transient-output
attenuation property, that is: sufficiently large
positive \(v\) drives the controller output \(w=z_1\) to an arbitrarily small
neighborhood of zero.

\begin{lemma}[The PITO property of the antithetic controller]
\label{lem:aic-pito}
The antithetic controller \eqref{eq:antithetic_controller} satisfies the
PITO property from the nonnegative input \(v\) to the output \(w=z_1\).
\end{lemma}

\begin{proof}
This proof will proceed as follows: first we show that large inputs $v$ will eventually drive $z_2$ to be ``large''. Then we will show that such large $z_2$ must eventually drive $z_1$ to small values, proving that it is PITO from $v$ to $z_1$.

By Lemma~\ref{lem:aic-positive-forward-complete}, $z_1(t),z_2(t)\ge0$ for all
$t\ge0$, for every initial condition $z(0)\in\mathbb R^2_{\ge0}$ and every
admissible input $v$. We construct explicit functions $V,T$ satisfying
Definition~\ref{def:PITO}.

Define
\[
    p(t):=\alpha_4z_1(t)-\alpha_2z_2(t)\le \alpha_4z_1(t).
\]
Along solutions of~\eqref{eq:antithetic_controller},
\[
    \dot p
    =
    \alpha_4\dot z_1-\alpha_2\dot z_2
    =
    \alpha_4(\alpha_1-\alpha_2z_1z_2)-\alpha_2(\alpha_3v-\alpha_4z_1z_2)
    =
    \alpha_1\alpha_4-\alpha_2\alpha_3v(t),
\]
since the bilinear terms $\alpha_2\alpha_4z_1z_2$ cancel exactly. In
particular, $\dot p$ does not depend on $z_1,z_2$ at all, only on $v$.

Fix
\[
    V_0:=\frac{2\alpha_1\alpha_4}{\alpha_2\alpha_3},
\]
and suppose $v(t)\ge V_0$ for all $t\ge0$. Integrating the expression for
$\dot p$ gives,
\[
    p(t)=p(0)+\alpha_1\alpha_4t-\alpha_2\alpha_3\int_0^tv(s)\,ds
    \le
    p(0)+\alpha_1\alpha_4t-\alpha_2\alpha_3V_0t
    =
    p(0)-\alpha_1\alpha_4t,
\]
by using $\int_0^tv(s)\,ds\ge V_0t$ and $\alpha_2\alpha_3V_0=2\alpha_1\alpha_4$.
Since $z_2(0)\ge0$, $p(0)=\alpha_4z_1(0)-\alpha_2z_2(0)\le\alpha_4z_1(0)$, so
\[
    p(t)\le\alpha_4z_1(0)-\alpha_1\alpha_4t
    \qquad
    \forall t\ge0.
\]
Rearranging $p(t)=\alpha_4z_1(t)-\alpha_2z_2(t)$ and using $z_1(t)\ge0$,
\[
    \alpha_2z_2(t)=\alpha_4z_1(t)-p(t)\ge-p(t)\ge\alpha_1\alpha_4t-\alpha_4z_1(0),
\]
so that
\begin{equation}
    \label{eq:z2bound}
    z_2(t)\ge\frac{\alpha_1\alpha_4}{\alpha_2}\,t-\frac{\alpha_4}{\alpha_2}z_1(0)
    \qquad
    \forall t\ge0.
\end{equation}
This bound holds for \emph{every} $z_2(0)\ge0$, since the actual value of
$z_2(0)$ was discarded (it can only help/tighten the inequality) precisely so that the resulting
bound depends only on $z_1(0)=w(0)$.

Next, fix $\varepsilon>0$ and set
\[
    M(\varepsilon):=\frac{2\alpha_1}{\alpha_2\varepsilon}.
\]
By \eqref{eq:z2bound}, the right-hand side reaches $M(\varepsilon)$ at
\[
    t_1:=t_1(\varepsilon,z_1(0))
    :=
    \frac{2}{\alpha_4\varepsilon}+\frac{z_1(0)}{\alpha_1},
\]
and, since the right-hand side of \eqref{eq:z2bound} is nondecreasing in $t$,
\[
    z_2(t)\ge M(\varepsilon)
    \qquad
    \forall t\ge t_1,
\]

showing that persistently large inputs will eventually drive $z_2$ to large enough values. 

Then, for $t\ge t_1$, using
$z_2(t)\ge M(\varepsilon)$ and $z_1(t)\ge0$,
\[
    \dot z_1(t)=\alpha_1-\alpha_2z_1(t)z_2(t)\le\alpha_1-\alpha_2M(\varepsilon)z_1(t).
\]

Next, since $\dot z_1\le\alpha_1$, then $z_1(t_1)\le z_1(0)+\alpha_1t_1=:Z_1$. By Gronwall's inequality, $z_1(t)\le Z(t)$ for all $t\ge t_1$, where $Z$
solves the corresponding linear equation $\dot Z=\alpha_1-\alpha_2M(\varepsilon)Z$,
$Z(t_1)=Z_1$. Explicitly, for $t\ge t_1$, write $r:=\alpha_2M(\varepsilon)$ for brevity, so that $Z$ solves
the linear initial value problem
\[
    \dot Z(t)=\alpha_1-rZ(t),
    \qquad
    Z(t_1)=Z_1 .
\]
Multiplying both sides by the integrating factor $e^{r(t-t_1)}$,
\[
    e^{r(t-t_1)}\dot Z(t)+re^{r(t-t_1)}Z(t)=\alpha_1e^{r(t-t_1)}
    \quad\Longleftrightarrow\quad
    \frac{d}{dt}\Big[e^{r(t-t_1)}Z(t)\Big]=\alpha_1e^{r(t-t_1)} .
\]
Integrating from $t_1$ to $t$,
\[
    e^{r(t-t_1)}Z(t)-Z_1
    =
    \alpha_1\int_{t_1}^te^{r(s-t_1)}\,ds
    =
    \frac{\alpha_1}{r}\Big(e^{r(t-t_1)}-1\Big) .
\]
Dividing by $e^{r(t-t_1)}$ gives the closed-form solution
\[
    Z(t)
    =
    Z_1e^{-r(t-t_1)}+\frac{\alpha_1}{r}\Big(1-e^{-r(t-t_1)}\Big)
    =
    \frac{\alpha_1}{r}+\left(Z_1-\frac{\alpha_1}{r}\right)e^{-r(t-t_1)},
    \qquad t\ge t_1 .
\]
Substituting $r=\alpha_2M(\varepsilon)$ back in recovers

\[
    Z(t)
    =
    \frac{\alpha_1}{\alpha_2M(\varepsilon)}
    +
    \left(Z_1-\frac{\alpha_1}{\alpha_2M(\varepsilon)}\right)e^{-\alpha_2M(\varepsilon)(t-t_1)}
    \le
    \frac{\alpha_1}{\alpha_2M(\varepsilon)}+Z_1e^{-\alpha_2M(\varepsilon)(t-t_1)},
\]
the last step using $\alpha_1/(\alpha_2M(\varepsilon))\ge0$. By the choice of
$M(\varepsilon)$,
\[
    \frac{\alpha_1}{\alpha_2M(\varepsilon)}=\frac\varepsilon2 ,
\]
so
\[
    z_1(t)\le\frac\varepsilon2+Z_1e^{-\alpha_2M(\varepsilon)(t-t_1)}
    \qquad
    \forall t\ge t_1 .
\]
The right-hand side is $\le\varepsilon$ once
$Z_1e^{-\alpha_2M(\varepsilon)(t-t_1)}\le\varepsilon/2$, i.e., once
\[
    t-t_1\ge\frac{1}{\alpha_2M(\varepsilon)}\ln\!\left(\frac{2Z_1}{\varepsilon}\right)
    =
    \frac{\varepsilon}{2\alpha_1}\ln\!\left(\frac{2Z_1}{\varepsilon}\right).
\]

\emph{Conclusion.} Define
\[
    V(\varepsilon):=V_0=\frac{2\alpha_1\alpha_4}{\alpha_2\alpha_3},
    \qquad
    T(\varepsilon,z_1(0))
    :=
    t_1(\varepsilon,z_1(0))
    +
    \max\left\{0,\ \frac{\varepsilon}{2\alpha_1}\ln\!\left(\frac{2Z_1}{\varepsilon}\right)\right\},
\]
with $Z_1=z_1(0)+\alpha_1t_1(\varepsilon,z_1(0))$ as above. By construction,
$v(t)\ge V(\varepsilon)$ for all $t\ge0$ implies $w(t)=z_1(t)\le\varepsilon$
for all $t\ge T(\varepsilon,z_1(0))$, for every $z(0)\in\mathbb R^2_{\ge0}$.
Both $t_1(\varepsilon,\cdot)$ and $Z_1$ are nondecreasing in $z_1(0)$, and the
logarithm is nondecreasing in $Z_1$, so $T(\varepsilon,\cdot)$ is
nondecreasing, as Definition~\ref{def:PITO} requires, and concluding the proof.
\end{proof}

From the results in Lemmas~\ref{lem:aic-positive-forward-complete} and~\ref{lem:aic-pito} we can state the following corollary of Theorem~\ref{thm:no_windup}.

\begin{corollary}[Bounded control signal for AIC interconnections]
\label{cor:aic-no-windup}
Consider the closed-loop interconnection between the plant
\eqref{eq:x-subsystem} and the antithetic integral controller
\eqref{eq:antithetic_controller}, given by
\[
    u=w=z_1,
    \qquad
    y=h(x).
\]
Assume that the closed-loop interconnection is forward complete, and that
the plant is positive and satisfies the PIPO property. Then, for every initial condition
\[
    x(0)\in\mathbb R^n_{\ge0},
    \qquad
    z(0)\in\mathbb R^2_{\ge0},
\]
the controller output
\[
    u(t)=w(t)=z_1(t)
\]
is bounded on \([0,\infty)\).

If, in addition, the plant is BIBS stable with respect to the input \(u\), then
the plant state \(x(t)\) is bounded on \([0,\infty)\).
\end{corollary}

\begin{proof}
By Lemma~\ref{lem:aic-positive-forward-complete}, the positive cone
\(\mathbb R^2_{\ge0}\) is forward invariant for the antithetic controller.
Hence
\[
    z_1(t)\ge0,
    \qquad
    z_2(t)\ge0,
    \qquad
    \forall t\ge0.
\]
In particular, the plant input
\[
    u(t)=z_1(t)
\]
is nonnegative for all \(t\ge0\). Since the plant is positive, its output is
nonnegative as well:
\[
    y(t)=h(x(t))\ge0.
\]

We prove boundedness of \(z_1(t)\) by contradiction. Fix
\(\varepsilon>0\). By Lemma~\ref{lem:aic-pito}, there exists
\(V_\varepsilon:=V(\varepsilon)>0\) such that the antithetic controller satisfies the
one-sided PITO property whenever
\[
    v(t)\ge V_\varepsilon.
\]
By the PIPO property of the plant, there exist
\[
    U_\varepsilon:=U(V_\varepsilon),
    \qquad
    T_\varepsilon^x:=T(V_\varepsilon),
\]
such that, for every \(\overline T\ge T_\varepsilon^x\),
\[
    u(t)\ge U_\varepsilon
    \quad \forall t\ge0
    \implies
    y(\overline T)\ge V_\varepsilon .
\]
By time invariance, the same implication holds on shifted intervals: if
\[
    u(t)\ge U_\varepsilon
    \qquad
    \forall t\in[a,b],
    \qquad
    b-a\ge T_\varepsilon^x,
\]
then
\[
    y(t)\ge V_\varepsilon
    \qquad
    \forall t\in[a+T_\varepsilon^x,b].
\]

Since \(w=z_1\), the antithetic controller satisfies the one-sided growth bound
\[
    \dot w(t)=\dot z_1(t)
    =
    \alpha_1-\alpha_2z_1(t)z_2(t)
    \le
    \alpha_1,
\]
because \(z_1(t),z_2(t)\ge0\). Therefore, for every \(t_2\ge t_1\ge0\),
\[
    w(t_2)\le w(t_1)+\alpha_1(t_2-t_1).
\]

Choose \(\eta>0\), and define
\[
    A:=\max\{U_\varepsilon,z_1(0),\varepsilon+\eta\}.
\]
Set
\[
    \delta:=A+\alpha_1T_\varepsilon^x.
\]
Let
\[
    T_\varepsilon^z:=T(\varepsilon,\delta)
\]
be the PITO time from Lemma~\ref{lem:aic-pito}. Finally, define
\[
    R:=A+\alpha_1\left(T_\varepsilon^x+T_\varepsilon^z\right).
\]

We claim that
\[
    z_1(t)\le R
    \qquad
    \forall t\ge0.
\]
Suppose, by contradiction, that this is false. Let \(\tau>0\) be the first time
such that
\[
    z_1(\tau)=R,
    \qquad
    z_1(t)\le R
    \quad
    \forall t<\tau.
\]
Since \(R>A\ge z_1(0)\), there exists a last time
\(\sigma\in[0,\tau)\) such that
\[
    z_1(\sigma)=A.
\]
By the definition of \(\sigma\),
\[
    z_1(t)\ge A
    \qquad
    \forall t\in[\sigma,\tau].
\]
Since \(A\ge U_\varepsilon\) and \(u=z_1\), it follows that
\[
    u(t)\ge U_\varepsilon
    \qquad
    \forall t\in[\sigma,\tau].
\]

Using the one-sided growth estimate, we obtain
\[
    R
    =
    z_1(\tau)
    \le
    z_1(\sigma)+\alpha_1(\tau-\sigma)
    =
    A+\alpha_1(\tau-\sigma).
\]
Therefore,
\[
    \tau-\sigma
    \ge
    T_\varepsilon^x+T_\varepsilon^z.
\]
Define
\[
    s:=\sigma+T_\varepsilon^x.
\]
Then \(s+T_\varepsilon^z\le\tau\). By the shifted PIPO implication,
\[
    y(t)\ge V_\varepsilon
    \qquad
    \forall t\in[s,s+T_\varepsilon^z].
\]

Moreover, using again the one-sided growth estimate,
\[
    z_1(s)
    \le
    z_1(\sigma)+\alpha_1T_\varepsilon^x
    =
    A+\alpha_1T_\varepsilon^x
    =
    \delta.
\]
Thus the hypotheses of Lemma~\ref{lem:aic-pito} are satisfied on the shifted
interval starting at \(s\). Hence
\[
    z_1(s+T_\varepsilon^z)\le\varepsilon.
\]
On the other hand,
\[
    s+T_\varepsilon^z\in[\sigma,\tau],
\]
and \(z_1(t)\ge A\) for every \(t\in[\sigma,\tau]\). Therefore
\[
    z_1(s+T_\varepsilon^z)\ge A>\varepsilon,
\]
which is a contradiction. Hence
\[
    z_1(t)\le R
    \qquad
    \forall t\ge0.
\]
Therefore the control signal
\[
    u(t)=w(t)=z_1(t)
\]
is bounded on \([0,\infty)\).

Finally, suppose that the plant is BIBS stable with respect to \(u\). Since
\(u(t)=z_1(t)\) is bounded and nonnegative, BIBS stability implies that the
plant state \(x(t)\) is bounded on \([0,\infty)\).
\end{proof}
Notice that Corollary~\ref{cor:aic-no-windup} should not be interpreted as boundedness of the full antithetic controller state. The result guarantees boundedness of the control signal
\[
    u(t)=w(t)=z_1(t),
\]
and, under BIBS stability of the plant, boundedness of the plant state \(x(t)\).
However, boundedness of the second antithetic species \(z_2(t)\) does not
follow directly from these facts.

Indeed, the \(z_2\)-dynamics are
\[
    \dot z_2
    =
    \alpha_3 y-\alpha_4 z_1z_2.
\]
Thus \(z_2\) is produced by the plant output \(y\) and degraded through the
bilinear annihilation term \(z_1z_2\). Even if \(z_1\) and \(x\) are bounded,
one only obtains boundedness of \(y=h(x)\). This does not by itself rule out
growth of \(z_2\), because the damping coefficient multiplying \(z_2\) is
\(\alpha_4 z_1(t)\), which may become arbitrarily small along the trajectory.
In particular, boundedness of \(z_1\) gives no positive lower bound on \(z_1\).

Therefore, additional structure is generally needed to conclude boundedness of
the full closed-loop state \((x,z_1,z_2)\). Such structure could take the form
of dissipativity, persistence, or detectability-type assumptions ensuring that
large values of \(z_2\) force sufficient decay in \(z_1\), or that the plant
output \(y\) cannot continue feeding the \(z_2\)-equation without eventually
driving the feedback loop into a regime where \(z_2\) is dissipated. For
example, in concrete positive-system models, one may prove boundedness of
\(z_2\) by exploiting dissipativity or finite-gain properties of the plant
together with the special bilinear structure of the antithetic controller. Such
model-dependent arguments go beyond the input-output PIPO/PITO mechanism used
in Corollary~\ref{cor:aic-no-windup}.

\subsection{PIPO, PITO, and Boundedness for a Nonlinear Integral-Feedback Loop}
\label{sec:nonlinear-ii-split}
For another example, we revisit the ``nonlinear II'' integral-feedback
motif of~\cite{shoval2011},
\[
    \dot \xsas=\alpha \xsas(y_0-\ysas),
    \qquad
    \dot \ysas=\beta u\xsas-\gamma \ysas,
\]
but now decompose it explicitly into a controller and a plant, interpreting the pair $(\xsas,\ysas)$ as the interconnection.
We hold the external input at a fixed constant $u\equiv\bar u>0$, treating it
as a parameter rather than as a signal to be tested for persistence; the
memory species $\xsas$ becomes the controller, and the response
species $\ysas$ becomes the plant. To avoid confusing with the
notation $x,y$ already used for a generic plant in~\eqref{eq:x-subsystem},
we write the controller and plant using the notation established in this paper, that is:
\begin{subequations}
\label{eq:nonlinear-II-split}
\begin{align}
C:\quad
&\begin{cases}
    \dot z=\alpha z(y_0-v),\\
    w=z,\\
    z(0)=z_0\in\mathbb R_{\ge0},
\end{cases}
\label{eq:nonlinear-II-controller}\\
P_{\bar u}:\quad
&\begin{cases}
    \dot x=\beta\bar u\,w-\gamma x,\\
    y=x,\\
    x(0)=x_0\in\mathbb R_{\ge0},
\end{cases}
\label{eq:nonlinear-II-plant}
\end{align}
\end{subequations}
with $\alpha,\beta,\gamma>0$ and $y_0>0$ fixed, interconnected as usual by
$u=w=z$ and $v=y=x$. Under the identification
$z\leftrightarrow \xsas$, $x\leftrightarrow \ysas$,
the closed loop of~\eqref{eq:nonlinear-II-split} is exactly the original
nonlinear II system with $u\equiv\bar u$; we use this correspondence below to
invoke the results of~\cite{shoval2011} directly.
\begin{lemma}[Positivity and forward completeness]
\label{lem:nII-split-positive-complete}
For every locally essentially bounded $v:\mathbb R_{\ge0}\to\mathbb R_{\ge0}$
and every $z_0\ge0$, the solution of~\eqref{eq:nonlinear-II-controller}
is forward complete, nonnegative, and given explicitly by
\[
    z(t)=z_0\exp\left(\alpha\int_0^t(y_0-v(s))\,ds\right)
    \qquad
    \forall t\ge0.
\]
For every locally essentially bounded $w:\mathbb R_{\ge0}\to\mathbb R_{\ge0}$
and every $x_0\ge0$, the solution of~\eqref{eq:nonlinear-II-plant} is forward
complete, nonnegative, and given explicitly by
\[
    x(t)=x_0e^{-\gamma t}+\beta\bar u\int_0^te^{-\gamma(t-s)}w(s)\,ds
    \qquad
    \forall t\ge0.
\]
\end{lemma}
\begin{proof}
The right-hand side of~\eqref{eq:nonlinear-II-controller} is linear and
homogeneous in $z$ for any fixed measurable $v(\cdot)$, so the displayed
formula is the unique solution (local Lipschitzness gives uniqueness); since
$v$ is locally essentially bounded, $\int_0^t(y_0-v(s))\,ds$ is finite for
every finite $t$, so $z(t)$ is finite and nonnegative for all $t\ge0$, with no finite escape time.
The right-hand side of~\eqref{eq:nonlinear-II-plant} is affine in $x$ for any
fixed measurable $w(\cdot)$, so the displayed variation-of-constants formula
is the unique solution; since $w$ is locally essentially bounded, the
convolution integral is finite for every finite $t$, and manifestly
nonnegative since $x_0,\bar u,\beta,w(\cdot)\ge0$. Hence $x(t)\ge0$ and finite
for all $t\ge0$, with no finite escape time.
\end{proof}
With positivity and forward completeness established, we next prove that $P$ is PIPO and $C$ is PITO.
\begin{lemma}[PIPO property of $P_{\bar u}$]
\label{lem:nII-plant-pipo}
For every fixed $\bar u>0$, the plant $P_{\bar u}$ in~\eqref{eq:nonlinear-II-plant}
satisfies the PIPO property, with
\[
    U(\kappa):=\frac{2\gamma\kappa}{\beta\bar u},
    \qquad
    T(\kappa):=\frac{\ln2}{\gamma}.
\]
\end{lemma}
\begin{proof}
Fix $\kappa\ge0$ (the case $\kappa=0$ is trivial) and $x_0\ge0$, and let $w$
satisfy $w(t)\ge U(\kappa)$ for all $t\ge0$. By
Lemma~\ref{lem:nII-split-positive-complete},
\[
    y(t)=x(t)
    =
    x_0e^{-\gamma t}+\beta\bar u\int_0^te^{-\gamma(t-s)}w(s)\,ds
    \ge
    \beta\bar uU(\kappa)\int_0^te^{-\gamma(t-s)}\,ds
    =
    \beta\bar uU(\kappa)\cdot\frac{1-e^{-\gamma t}}{\gamma},
\]
using $x_0\ge0$ and $w(s)\ge U(\kappa)$ for all $s\ge0$. Substituting
$U(\kappa)=2\gamma\kappa/(\beta\bar u)$,
\[
    y(t)\ge2\kappa\left(1-e^{-\gamma t}\right).
\]
For $t\ge T(\kappa)=\ln(2)/\gamma$, $1-e^{-\gamma t}\ge1/2$, so $y(t)\ge\kappa$.
The bound above holds uniformly over $x_0\ge0$, so the PIPO implication holds
for every initial condition, as Definition~\ref{def:PIPO} requires.
\end{proof}

Notice that Lemma~\ref{lem:nII-plant-pipo} is a concrete instance of
Lemma~\ref{lem:monotone-pipo}: $P_{\bar u}$ trivially satisfies
Assumption~\ref{ass:monotone} (it is linear, hence monotone, with $h=\mathrm{id}$
order preserving) and Assumption~\ref{ass:characteristic}, with characteristic
$k_x(\bar w)=k_y(\bar w)=\beta\bar u\bar w/\gamma\in\mathcal K_\infty$.
Lemma~\ref{lem:nII-plant-pipo} simply exhibits explicit gains for this
specific linear instance, in place of the general, non-constructive argument
of Lemma~\ref{lem:monotone-pipo}; the same computation also shows $P_{\bar u}$
is BIBS stable, since for $0\le w(t)\le M_w$ the same variation-of-constants
formula gives $x(t)\le x_0+\beta\bar uM_w/\gamma$ for all $t\ge0$.

For the ``controller'' consider:

\begin{lemma}[PITO property of $C$]
\label{lem:nII-controller-pito}
The controller $C$ in~\eqref{eq:nonlinear-II-controller} satisfies the
PITO property from $v$ to $w=z$, with
\[
    V(\varepsilon):=y_0+1,
    \qquad
    T(\varepsilon,z_0):=
    \begin{cases}
        \dfrac{1}{\alpha}\ln\!\left(\dfrac{z_0}{\varepsilon}\right), & z_0>\varepsilon,\\[2mm]
        0, & z_0\le\varepsilon.
    \end{cases}
\]
\end{lemma}
\begin{proof}
Fix $\varepsilon>0$ and $z_0\ge0$, and let $v$ satisfy $v(t)\ge V(\varepsilon)=y_0+1$ for all $t\ge0$.
By Lemma~\ref{lem:nII-split-positive-complete},
\[
    z(t)
    =
    z_0\exp\left(\alpha\int_0^t(y_0-v(s))\,ds\right)
    \le
    z_0\exp\left(\alpha\int_0^t(y_0-V(\varepsilon))\,ds\right)
    =
    z_0e^{-\alpha t},
\]
by using $y_0-v(s)\le y_0-V(\varepsilon)=-1$ for all $s\ge0$. If $z_0\le\varepsilon$,
then $z(t)\le z_0\le\varepsilon$ for all $t\ge0$, and $T(\varepsilon,z_0)=0$
works. If $z_0>\varepsilon$, then $z(t)\le\varepsilon$ once
$z_0e^{-\alpha t}\le\varepsilon$, i.e., once
$t\ge\frac1\alpha\ln(z_0/\varepsilon)=T(\varepsilon,z_0)$. In either
case, $w(t)=z(t)\le\varepsilon$ for all $t\ge T(\varepsilon,z_0)$, and
$T(\varepsilon,\cdot)$ is (weakly) increasing in $z_0$, as
Definition~\ref{def:PITO} requires.
\end{proof}

Notice that, unlike the antithetic controller, whose output satisfies the purely constant
bound $\dot w\le\alpha_1$, the controller $C$ satisfies
\[
    \dot z=\alpha z(y_0-v)\le\alpha y_0\,z
    \qquad
    \forall t\ge0,
\]
since $v\ge0$. This is exactly the one-sided affine growth bound
$\dot w(t)\le\overline kw(t)+\omega$ assumed for Theorem~\ref{thm:no_windup}, with
$\overline k=\alpha y_0$ and $\omega=0$. 

With this, we are ready to look at the interconnection:

\begin{corollary}[Bounded control signal for the nonlinear-II loop]
\label{cor:nII-split-bounded}
Fix $\bar u>0$. For every initial condition $z_0\ge0$, $x_0\ge0$, the
control signal $u(t)=w(t)=z(t)$ in the closed-loop interconnection of
$C$~\eqref{eq:nonlinear-II-controller} and
$P_{\bar u}$~\eqref{eq:nonlinear-II-plant} is bounded on $[0,\infty)$.
Moreover, since $P_{\bar u}$ is BIBS stable, the plant state $x(t)$ is also bounded on
$[0,\infty)$.
\end{corollary}
\begin{proof}
By Lemma~\ref{lem:nII-split-positive-complete}, both subsystems are positive
and forward complete for arbitrary admissible inputs; the closed-loop
interconnection itself is forward complete because the Lyapunov function
constructed in the proof of Lemma~5.1 in~\cite{shoval2011} is proper, so
trajectories remain in a compact sublevel set for all forward time. By
Lemma~\ref{lem:nII-plant-pipo}, $P_{\bar u}$ satisfies the PIPO property. By
Lemma~\ref{lem:nII-controller-pito}, $C$ satisfies the PITO property. We have also established that the controller output satisfies the
one-sided affine growth bound of Theorem~\ref{thm:no_windup} with
$\overline k=\alpha y_0$, $\omega=0$. All hypotheses of
Theorem~\ref{thm:no_windup} are therefore satisfied, and the control signal
$u(t)=w(t)=z(t)$ is bounded on $[0,\infty)$.

Since $w=z$ is bounded, and $P_{\bar u}$ is BIBS stable with respect to $w$,
the plant state $x(t)$ is bounded on $[0,\infty)$ as well.
\end{proof}

Note that Corollary~\ref{cor:nII-split-bounded} recovers, through the PIPO/PITO framework, a conclusion already implied by Corollary~5.2
of~\cite{shoval2011} (global asymptotic stability of the nonlinear II system
for every fixed $\bar u>0$). The value of the argument above is, therefore, as a worked example of
Theorem~\ref{thm:no_windup} under a genuinely multiplicative growth bound
($\overline k=\alpha y_0>0$, $\omega=0$), complementing the purely additive
bound ($\overline k=0$, $\omega=\alpha_1$) illustrated by the antithetic
controller.

\section{Conclusions}
\label{sec:conclusion}

In this paper we studied the interconnection of two positive nonlinear control systems in a classical plant/controller feedback loop, and asked when the resulting closed-loop control signal is guaranteed to remain bounded. This is an input-output question in the spirit of anti-windup analysis, but posed
without linearity, without a specific controller architecture in mind, and without assuming stability of either subsystem in isolation.

To address this problem, we introduced two complementary input-output properties. The persistent-input/persistent-output (PIPO) property, required of the plant, formalizes the idea that a persistently large input eventually forces a persistently large output, uniformly over initial conditions. The persistent-input/transient-output (PITO) property, required of the controller, formalizes the complementary idea that a persistently large input eventually attenuates the controller's own output back down to a proportionally small level. Theorem~\ref{thm:no_windup} shows that these two properties (together with mild regularity assumptions on the systems) are jointly sufficient to guarantee boundedness of the control signal, regardless of initial condition. Corollary~\ref{cor:PIPOTO_BIBS} strengthens this to boundedness of the full plant state whenever the plant is additionally BIBS stable with respect to the control input. Because both properties are stated purely in terms of input-output behavior, the result applies uniformly across very different system architectures, without
requiring a common Lyapunov function or a shared notion of equilibrium for the interconnection.

We then showed that both properties are satisfied by broad and practically relevant classes of systems. On the plant side, Lemma~\ref{lem:monotone-pipo} shows that any positive stable monotone system whose input-output characteristic is of class $\mathcal K_\infty$ is automatically both PIPO and BIBS stable, so persistent amplification follows from a mild monotonicity and steady-state growth condition rather than from an ad hoc construction. On the controller side, we studied two examples from synthetic and systems biology. For the antithetic integral controller, after establishing its positivity and forward completeness (Lemma~\ref{lem:aic-positive-forward-complete}), we proved directly that it satisfies the PITO property (Lemma~\ref{lem:aic-pito}), and used this to conclude boundedness of the control signal for any positive PIPO plant it is connected to (Corollary~\ref{cor:aic-no-windup}). For the nonlinear-II integral-feedback motif of~\cite{shoval2011} (Lemmas~\ref{lem:nII-split-positive-complete}--\ref{lem:nII-controller-pito}), we decomposed the system into a log-linear controller and a linear plant. We showed that the controller and plant each satisfy the PIPO and PITO properties respectively, and with both gains available in closed form enabling us to certify boundedness of the states in this case (Corollary~\ref{cor:nII-split-bounded}). Together, these results show that the PIPO/PITO framework is not merely an abstract input-output pairing but one that can be certified constructively, with explicit gain functions, for systems of genuine interest in the literature.

Taken together, these results make the case for the PIPO/PITO framework as a general-purpose, input-output criterion for ruling out windup in positive interconnections. Its strength lies in reducing a question about the coupled nonlinear closed-loop dynamics to two comparatively simple properties, verified separately on the plant and the controller, through explicit and computable gain functions rather than a shared Lyapunov function or a common notion of equilibrium. 
\bibliographystyle{plain}
\bibliography{references}

@string{TAC="IEEE Trans.\ Automat.\ Control"}

@string{SIAM="SIAM J.\ Control Optim."}

@string{IEEE="Proc. IEEE"}

@string{s="Submitted"}

@article{Huang2018,
  author  = {Huang, H-H and Qian, Y and Del Vecchio, D},
  title   = {A quasi-integral controller for adaptation of genetic modules to variable ribosome demand},
  journal = {Nature Communications},
  volume  = {9},
  number  = {1},
  pages   = {5415},
  year    = {2018},
  doi     = {10.1038/s41467-018-07899-z},
  url     = {https://doi.org/10.1038/s41467-018-07899-z}
}

@book {mct,
    AUTHOR = {Sontag, E.D.},
 TITLE = {Mathematical {C}ontrol {T}heory. {D}eterministic {F}inite-{D}imensional {S}ystems},
    SERIES = {Texts in Applied Mathematics},
    VOLUME = {6},
   EDITION = {Second},
 PUBLISHER = {Springer-Verlag},
   ADDRESS = {New York},
      YEAR = {1998},
     PAGES = {xvi+531},
      ISBN = {0-387-98489-5},
      }

@INPROCEEDINGS{Margaliot-CDC,
  author={Margaliot, Michael and Wu, Chengshuai and Sontag, Eduardo D.},
  booktitle={2025 IEEE 64th Conference on Decision and Control (CDC)}, 
  title={Compact attractors of an antithetic integral feedback system have a simple structure}, 
  year={2025},
  volume={},
  number={},
  pages={2880-2885},
  doi={10.1109/CDC57313.2025.11312315}}

@article{rev_internal_model_2022,
   author = "Bin, Michelangelo and Huang, Jie and Isidori, Alberto and Marconi, Lorenzo and Mischiati, Matteo and Sontag, Eduardo",
   title = "Internal Models in Control, Bioengineering, and Neuroscience", 
   journal= "Annual Review of Control, Robotics, and Autonomous Systems",
   year = "2022",
   volume = "5",
   pages = "55-79",
  }

@article{AFC2016,
author={Briat, C. and  Gupta, A.  and  Khammash, M.}, title={Antithetic integral feedback ensures robust
perfect adaptation in noisy biomolecular networks},
journal={Cell Syst.},
volume={2}, pages={15-26}, year={2016},
}

@article{Khammash2019,
author={Stephanie K. Aoki and  Gabriele Lillacci and  Ankit Gupta and  Armin Baumschlager and  David Schweingruber and  Mustafa Khammash},
title={A universal biomolecular integral feedback controller for robust perfect adaptation},
journal={Nature}, volume={570}, pages={533-537},
year={2019},
}

@INPROCEEDINGS{AngeliS2004,
  author =      { D. Angeli and E. D. Sontag.},
  title =        {Interconnections of monotone systems with
steady-state characteristics},
  booktitle =    {Optimal control, stabilization and nonsmooth analysis},
  year =        {2004},
  editor=      {M. S. \mbox{de Queiroz} and M. A. Malisoff and P.
R. Wolenski},
  volume =      {301},
  series =      {Lecture Notes in Control and Inform. Sci.},
  pages =        {135-154},
  address =      {Berlin},
  publisher =    {Springer},
}

@article{mcs_angeli_2003,
title={Monotone control systems},
journal=TAC,
year={2003},
author={D. Angeli and E. D. Sontag},
volume={48},
pages={1684-1698},
}

@article{corentin_2020,
author = {Briat, Corentin},
title = {A Biology-Inspired Approach to the Positive Integral Control of Positive Systems: The Antithetic, Exponential, and Logistic Integral Controllers},
journal = {SIAM J. Applied Dynamical Systems},
volume = {19},
number = {1},
pages = {619-664},
year = {2020}
}

@article{Eyal_k_posi,
	title={A Generalization of Linear Positive Systems with Applications to Nonlinear Systems: Invariant Sets and the {Poincar\'{e}-Bendixson} Property},
	author={E. Weiss and M. Margaliot},
	year={2021}, volume = {123},
    pages = {109358},
	journal={Automatica},  
}

@inproceedings{agarwal_cdc2019,
  title={Some remarks on robust gene regulation in a biomolecular integral controller},
  author={D. K. Agrawal and R. Marshall and M. Ali Al-Radhawi and V. Noireaux and E. D. Sontag},
 booktitle = {Proc. 2019 IEEE Conf. Decision and Control},address={Nice,  France},
 year = {2019},
 pages = {2820-2825},
}

@article{agarwal2019naturecom,
  author = {D. K. Agrawal and R. Marshall and V. Noireaux and E. D. Sontag},
  title = {In vitro implementation of robust gene regulation in a synthetic biomolecular integral controller},
  journal={Nature Communications},
  year={2019},volume={10},pages={5760}, 
}

@INPROCEEDINGS{18cdc_tutorial_imp,
  title={Internal models in control, biology and neuroscience},
  author={J. Huang and A. Isidori and L. Marconi and M. Mischiati and E. D. Sontag and W. M. Wonham},
 booktitle = {Proc. 2018 IEEE Conf. Decision and Control},
 year = {2018},
 pages = {5370-5390},
}

@article{2019biorxiv_margaliot_sontag,
	author = {Margaliot, M. and Sontag, E.D.},
	title = {Compact attractors of an antithetic integral feedback system have a simple structure (preprint version)},
	note = {DOI: 10.1101/868000},
	elocation-id = {868000},
	year = {2019},
	doi = {10.1101/868000},
	publisher = {Cold Spring Harbor Laboratory},
	URL = {https://www.biorxiv.org/content/early/2019/12/08/868000},
	eprint = {https://www.biorxiv.org/content/early/2019/12/08/868000.full.pdf},
	journal = {bioRxiv}
}

@article{katz2025instability,
author  = {Katz, Rami and Giordano, Giulia and Margaliot, Michael},
title   = {Instability of equilibrium and convergence to periodic orbits in strongly 2-cooperative systems},
journal = {Journal of Differential Equations},
volume  = {444},
pages   = {113651},
year    = {2025},
doi     = {10.1016/j.jde.2025.113651},
url     = {https://doi.org/10.1016/j.jde.2025.113651}
}

@article{shoval2011,
  title={Symmetry invariance for adapting biological systems},
  author={Shoval, Oren and Alon, Uri and Sontag, Eduardo},
  journal={SIAM journal on applied dynamical systems},
  volume={10},
  number={3},
  pages={857--886},
  year={2011},
  publisher={SIAM}
}

@article{galeani2009tutorial,
  title={A tutorial on modern anti-windup design},
  author={Galeani, Sergio and Tarbouriech, Sophie and Turner, Matthew and Zaccarian, Luca},
  journal={European Journal of Control},
  volume={15},
  number={3-4},
  pages={418--440},
  year={2009},
  publisher={Elsevier}
}

@article{wafi-boundedness-antithetic,
  title={Boundedness of solutions in feedback systems with antithetic controllers},
  author={Wafi, Moh Kamalul and de Oliveira, Arthur CB and Sontag, Eduardo D},
  journal={arXiv preprint arXiv:2604.27290},
  year={2026}
}

\end{document}